\documentclass[11pt]{article}

\usepackage[letterpaper,margin=1in]{geometry}
\usepackage[utf8]{inputenc} 
\usepackage[T1]{fontenc}    
\usepackage[colorlinks=true]{hyperref}       
\usepackage{url}            
\usepackage{booktabs}       
\usepackage{amsfonts}       
\usepackage{amsmath,dsfont,amsthm,mathtools,amssymb}
\usepackage{subcaption}

\usepackage{pgfplots}
\pgfplotsset{compat=1.18}

\newcommand{\calP}{\mathcal{P}}

\newcommand{\calT}{\mathcal{T}}

\newcommand{\calX}{\mathcal{X}}
\newcommand{\calY}{\mathcal{Y}}

\newcommand{\xvec}{\boldsymbol{x}}
\newcommand{\yvec}{\boldsymbol{y}}

\newcommand{\Xvec}{\boldsymbol{X}}
\newcommand{\Yvec}{\boldsymbol{Y}}

\newcommand{\bbE}{\mathbb{E}}

\newcommand{\indic}{\mathds{1}}

\newcommand{\typical}[2]{\calT^{#1}_{#2}}
\newcommand{\tiltP}[2]{#1^{(#2)}}

\newcommand{\better}{\succ}
\newcommand{\bettereq}{\succeq}

\newtheorem{theorem}{Theorem}

\newtheorem{proposition}{Proposition}

\newtheorem{remark}{Remark}

\DeclareMathOperator*{\supp}{supp}

\hypersetup{
    colorlinks,
    linkcolor=blue,
    citecolor=red,
    urlcolor=blue
}

\usepackage[
backend=biber,
style=alphabetic,
maxbibnames=15,
maxalphanames=5,
minalphanames=3,
doi=false,
isbn=false,
eprint=false,
backref=true,
]{biblatex}
\title{Tail Exponents of Conditional Guesswork via the Method of Types}

\author{%
    Adway Girish$^*$, Andreina Patrizia Motter$^\dagger$, Emre Telatar$^*$ \\
{\small    $^*$School of Computer and Communication Sciences, EPFL, 
Switzerland}\\
{\small    $^\dagger$Federal Department of Foreign Affairs, Switzerland}\\
 {\small   \texttt{adway.girish@epfl.ch},\,\texttt{andreina.motter@alumni.epfl.ch},\,\texttt{emre.telatar@epfl.ch}}
}

\begin{document}

\maketitle

\begin{abstract}
We study the problem of guessing a realization of an i.i.d.~random 
sequence given element-wise correlated side-information.  
We use type-counting to provide estimates of the tail probabilities of 
the number of guesses for the case without side-information, which was 
shown earlier through large-deviation techniques. We then extend the 
same counting argument to the conditional setting, obtaining new 
explicit expressions for the corresponding guesswork exponents as 
divergences involving conditional tilted distributions. Finally, we 
provide an application of these exponents to brute-force password 
guessing with side-information.
\end{abstract}

\section{Introduction} \label{sec: intro}
Consider the problem of guessing a realization of a finite-alphabet random 
variable, first studied by Massey~\cite{massey}: $X$ is generated according 
to 
a probability distribution $P_X$, and the guesser is allowed to ask 
questions 
of the form ``Is $X = x$?'' for different choices of $x$, until they 
correctly 
guess the realization $X$.  The goal is to do so with ``as few guesses as 
possible''. If $P_X$ is known, the optimal strategy is to guess symbols in 
decreasing order of their probability. Let $G(x)$ denote the ``guesswork'', 
i.e., the number of attempts needed to correctly guess $x$.  If the guesser 
also has access to side-information $Y$ that is correlated with $X$, then 
the 
optimal strategy is to guess symbols in decreasing order of their 
conditional 
probability, and let $G(x|y)$ denote the (conditional) guesswork for $x$ on 
observing $y$.  When referring to guesswork in the remainder of the paper, 
we 
always consider the optimal strategy of guessing symbols from most likely 
to 
least likely.  Guesswork is a natural operational quantity in settings 
where 
one gets an immediate binary ``yes''/``no'' feedback, such as brute-force 
password guessing~\cite{brute_botnets,brute_decent,guess_storage}, 
guessing-noise decoding~\cite{duffy_grand},
computing lower bounds to the computational effort of sequential decoding 
(via an oracle that provides immediate feedback)~\cite{arikan}, and so on.

Now suppose that the object to be guessed is an $n$-length sequence $\Xvec$ 
that is correlated with side-information $\Yvec$ available to the guesser, 
where the entries $(X_i,Y_i)$ are generated i.i.d.\ from a distribution 
$P_{XY}$.
Then the guesswork grows exponentially in $n$, and Arikan~\cite{arikan} 
characterized the exponent of the moment $\bbE[G(\Xvec|\Yvec)^{\rho}]$ in 
terms of the conditional R\'enyi entropy of order $1/(1+\rho)$.  A natural 
quantity to study is therefore the tail probabilities of the guesswork, in 
particular, for what values of $h$ the quantities $\Pr\{G(\Xvec|\Yvec) < 
2^{nh}\}$ and $\Pr\{G(\Xvec|\Yvec) > 2^{nh}\}$ decay exponentially in $n 
\to \infty$ and with what exponent.  

Such large-deviation results have been explored previously in the 
literature. Christiansen and Duffy~\cite{duffy_ldp} showed for $\Xvec$ 
not 
necessarily i.i.d., that $n^{-1}\log G(\Xvec)$, without 
side-information, 
satisfies a large-deviations principle (LDP).  However, it requires an 
assumption on the scaled cumulant generating function (sCGF) that 
clearly holds 
for i.i.d.~sources, but is hard to verify in general.  For 
i.i.d.~$\Xvec$, 
Beirami et al.~\cite{beirami_types,beirami_ldp} explicitly 
characterized the 
exponents of $\Pr\{G(\Xvec)\gtrless 2^{nh}\}$ as a function of $h$, in 
terms of 
information-theoretic quantities involving ``tilted'' distributions.  
To do so, 
they developed a theory of ``weakly typical tilted sets'' that requires 
several 
technical results.
A similar LDP for the conditional case with general sources was 
established by Li~\cite{li_cond_ldp}, but again with an assumption on 
the sCGF similar to Christiansen and Duffy~\cite{duffy_ldp}, that is 
hard to verify for non-i.i.d.~sources.  Moreover, even for 
i.i.d.~sources, the resulting rate function is not explicitly evaluated 
as Beirami et al.~\cite{beirami_types,beirami_ldp} do for the 
unconditional case.

In this paper, we consider i.i.d.~sources and explicitly derive 
the exponents of $\Pr\{G(\Xvec|\Yvec)\gtrless 2^{nh}\}$ as a function of 
$h$ using the method of types. In addition to giving a short and direct 
proof, this gives the exponent immediately as the solution to an 
entropy-constrained divergence minimization problem and makes clear which 
empirical distributions of $(\Xvec,\Yvec)$ contribute the most to the tail 
event.  Finally, since the argument uses explicit type-class bounds, it can 
also be used to obtain finite-$n$ upper and lower bounds to the probability 
that match up to polynomial factors in $n$ without requiring approximation 
arguments, unlike the large-deviations machinery which only yields 
asymptotic limits~\cite{duffy_ldp,beirami_ldp,li_cond_ldp} (we do not 
provide these bounds here to keep our presentation concise).
We illustrate the use of these explicit characterizations through an 
application to brute-force password guessing with side-information.

\section{Preliminaries}
For completeness, we first provide a review of our main technical tool, the 
method of types~\cite{csiszar2011information,csiszar_mot}.
We then introduce the technical aspects of the guessing problem and define 
an ``order'' on these types, which will prove useful for the analysis in 
Section~\ref{sec: results}.

\subsection{Notation}
All logarithms and exponents are taken to base 2, i.e., we write $a = 2^b$ 
as $a = \exp(b)$ and $b = \log(a)$.
We use calligraphic letters $\calX,\calY$ to denote finite sets. 
Given $\calX$ of cardinality $|\calX|$, let $\calP(\calX)$ denote the set 
of 
all probability distributions on $\calX$. When $X$ is a random variable 
taking 
values on $\calX$ with distribution $P \in \calP(\calX)$, we use $X$ and 
$P$ 
interchangeably in expressions involving information measures.  The entropy 
of 
$X \sim P$ is denoted by $H(P) = H(X) = -\sum_x P(x)\log P(x)$.  For 
distributions $P,Q \in \calP(\calX)$, the KL divergence is $D(Q \| P) = 
\sum_x 
Q(x)\log\big(Q(x)/P(x)\big)$ and the cross entropy is $H(Q \| P) = H(Q) + 
D(Q 
\| P)$.  Given a joint distribution $P_{XY} \in \calP(\calX \times \calY)$, 
we 
also define the conditional entropy $H(X \mid Y) = -\sum_{x,y} P_{XY}(x,y) 
\log 
P_{X|Y}(x|y)$, which we also denote as $H_{P_{XY}}(X \mid Y)$ to make the 
joint 
distribution clear.  Also given a  distribution $Q_{XY}$, we define the 
conditional KL divergence $D(Q_{X|Y} \| P_{X|Y} | Q_Y) = \sum_{x,y} 
Q_{XY}(x,y)\log\big(Q_{X|Y}(x|y)/P_{X|Y}(x|y)\big)$, and the conditional 
cross 
entropy $H(Q_{X|Y} \| P_{X|Y} | Q_Y) = H_{Q_{XY}}(X|Y) + D(Q_{X|Y} \| 
P_{X|Y} | 
Q_Y)$.  We may drop subscripts to simplify notation when the distribution 
is 
clear from context. 
We use boldface letters to denote $n$-length vectors, such as $\xvec \in 
\calX^n$, which has components $x_1,\dots,x_n$.  We use $\supp(P)$ to 
denote the support of a distribution.

Given a distribution $P$ on $\calX$ and $\lambda > 0$, we define the 
\emph{tilted} distribution $\tiltP{P}{\lambda}$ as\vspace*{-8pt}
\begin{equation*}
	\tiltP{P}{\lambda}(x) = \frac{P(x)^{\lambda}}{\sum_{x'} 
		P(x')^{\lambda}}.
\end{equation*}
Similarly, given a conditional distribution $P_{X|Y}$ and $\lambda >0$, the 
\emph{conditional tilted} distribution $\tiltP{P_{X|Y}}{\lambda}$ is given 
by
\begin{equation*}
	\tiltP{P_{X|Y}}{\lambda}(x \mid y) = \frac{P_{X|Y}(x \mid 
		y)^{\lambda}}{\sum_{x'} P_{X|Y}(x' \mid y)^{\lambda}}.
\end{equation*}

\subsection{Method of types}
Given a sequence $\xvec \in \calX^n$, denote the \emph{type} of $\xvec$ by 
the probability distribution $Q_{\xvec} \in \calP(\calX)$ with 
$Q_{\xvec}(a)$ equal to the relative frequency of $a \in \calX$ among 
$x_1,\dots,x_n$.  Let $\calP_n(\calX)$ denote the set of all $n$-types, 
i.e., types of $n$-length sequences.  The number of $n$-types is at most 
$(n+1)^{|\calX|}$. 
Given a probability distribution $P \in \calP(\calX)$, let $\typical{n}{P} 
\coloneqq \{\xvec \in \calX^n : Q_{\xvec} = P\}$ denote the type class, the 
set of all $n$-length sequences of type $P$. The cardinality of the type 
class is bounded as $|\typical{n}{P}| \doteq \exp(n\, H(P))$, where 
$\doteq$ means that the left and right-hand sides are equal up to 
polynomial factors in $n$. For $\Xvec$ generated i.i.d.\ according to $P$, 
we have $\Pr\{\Xvec = \xvec\} = \exp[-n\,H(Q_{\xvec} \| P)]$. 

Similarly define the joint type $Q_{\xvec\yvec} \in \calP(\calX \times 
\calY)$ and for a distribution $P_{XY} \in \calP(\calX \times \calY)$, 
define the joint type class $\typical{n}{P_{XY}}$.  Given $\yvec \in 
\calY^n$ and a conditional distribution $V$, define the \emph{conditional 
	type class} to be $\typical{n}{V}(\yvec) \coloneqq \{\xvec \in \calX^n : 
(\xvec, \yvec) \in \typical{n}{Q_{\yvec}V}\}$. The cardinality of the 
conditional type class is bounded as $|\typical{n}{P_{X|Y}}(\yvec)| \doteq 
\exp(n\, H(X|Y))$, where the conditional entropy $H(X|Y)$ is computed 
according to the joint distribution induced by the type $Q_{\yvec}$ and 
$P_{X|Y}$. For $\Xvec$ generated from $\yvec$ conditionally independently 
according to the conditional distribution $P_{X|Y}$, we have $\Pr\{\Xvec = 
\xvec \mid \Yvec = \yvec\} = \exp[-n\,H(Q_{\xvec | \yvec} \| P_{X|Y} | 
Q_{\yvec})]$ where $Q_{\xvec|\yvec}$ is the conditional type induced by 
$(\xvec,\yvec)$, given by $Q_{\xvec|\yvec}(x|y) = Q_{\xvec\yvec}(x,y) / 
Q_{\yvec}(y)$.  

\subsection{Guessing and orders on types}  \label{sec: guess_order}
Suppose that the object to be guessed is $\Xvec = X_1,\dots,X_n$ generated 
i.i.d.~from a distribution $P \in \calP(\calX)$.  Then, we guess $\xvec \in 
\calX^n$ before $\xvec'$ if $P_{\Xvec}(\xvec) \geq P_{\Xvec}(\xvec')$, with 
ties broken arbitrarily.  Note that $P_{\Xvec}(\xvec) = 
\exp[-n\,H(Q_{\xvec} \| P)]$, and hence, define the total order $\bettereq$ 
between types as $Q \bettereq Q'$ to mean $H(Q \| P) \leq H(Q' \| P)$, 
i.e., 
\begin{equation*}
	\sum_{x} [Q(x) - Q'(x)] \log\frac1{P(x)}  \leq 0.
\end{equation*}
Similarly, define the strict order $\better$ if the above inequality is 
strict. Note that $\better$ and $\bettereq$ depend on $P$.  
The optimal guessing strategy is to go through types in decreasing order 
(picking arbitrarily among $Q$ and $Q'$ if $H(Q \| P) = H(Q' \| P)$) and 
guessing sequences within each type arbitrarily. Recall that $G(\xvec)$ 
denotes the number of guesses needed (guesswork) to correctly guess $\xvec$ 
with this strategy.  

Now consider a conditional version of the problem, where we have access to 
$\Yvec$ that is correlated element-wise with $\Xvec$, i.e., $(X_i,Y_i)$ is 
generated i.i.d.~from a distribution $P_{XY} \in \calP(\calX \times 
\calY)$.  Having observed $\yvec$, we guess $\xvec$ before $\xvec'$ if 
$P_{\Xvec|\Yvec}(\xvec|\yvec) \geq P_{\Xvec|\Yvec}(\xvec'|\yvec)$.  Note 
that $P_{\Xvec|\Yvec}(\xvec|\yvec) = \exp[-n\,H(Q_{X|Y} \| P_{X|Y} \mid 
Q_{\yvec})]$, and hence, we define the total order between conditional 
types as $Q_{X|Y} \bettereq Q'_{X|Y}$ to mean $H(Q_{X|Y} \| P_{X|Y} \mid 
Q_{\yvec}) \leq H(Q'_{X|Y} \| P_{X|Y} \mid Q_{\yvec})$, i.e., 
\begin{equation*}
	\sum_{x,y} [Q_{XY}(x,y) - Q_{XY}'(x,y)] \log\frac1{P_{X|Y}(x|y)}  \leq 
	0,
\end{equation*}  
where $Q_{XY}(x,y) = Q_{X|Y}(x|y)Q_{\yvec}(y)$ is the joint distribution 
induced by the conditional type $Q_{X|Y}$ and the type $Q_{\yvec}$.  Also 
define the strict order similarly, and note that $\better$ and $\bettereq$ 
depend on $P_{X|Y}$ and $Q_{\yvec}$.
This order on conditional types also induces a partial order on joint types 
in 
the following way: we write $Q_{XY} \bettereq Q'_{XY}$ to mean that they 
have 
the same $Y$-marginals, i.e., $Q_Y = Q'_Y$, and their conditional types 
satisfy 
$Q_{X|Y} \bettereq Q'_{X|Y}$.  Recall that $G(\xvec | \yvec)$ denotes the 
guesswork to correctly guess $\xvec$ having observed $\yvec$ with the 
optimal 
strategy, i.e., guessing sequences from type $Q$ before $Q'$ if $Q 
\bettereq 
Q'$.

\section{Tail exponents of guesswork} \label{sec: results}

For a pair $(\Xvec,\Yvec)$ with entries $(X_i, Y_i)$ generated i.i.d.~from 
a distribution $P_{XY} \in \calP(\calX \times \calY)$, we are interested in 
the tail exponents of conditionally guessing $\Xvec$ from $\Yvec$, i.e.,
\begin{equation} \label{eqn: cond_exponents}
	e_{X|Y}^{\gtrless}(h) \coloneqq  \lim_{n \to \infty} -\tfrac1n \log 
	\Pr\{G(\Xvec | \Yvec) \gtrless 2^{nh}\},
\end{equation}
where $\gtrless$ can be either $>$ or $<$.  Also suppose without loss of 
generality that $P_X$ and $P_{Y}$ have full support.  In addition, we 
assume that $P_X$ and each $P_{X|Y}(\cdot | y)$ have unique maximizers. 
Analogous results can be obtained even without this assumption, but this 
simplifies the presentation of our results.

The unconditional version of this problem, i.e., the exponents of 
$\Pr\{G(\Xvec)\gtrless 2^{nh}\}$, was explicitly characterized by Beirami 
et 
al.~\cite{beirami_ldp}.  We first recover their result directly using a 
simple 
type-counting argument, which is well-suited when the sets are finite and 
the 
distribution is i.i.d., as in our setting (and indeed, also the setting 
that 
Beirami et al.~\cite{beirami_ldp} consider) and then apply the 
same approach to conditional guesswork.  Li~\cite{li_cond_ldp} similarly 
to~\cite{beirami_ldp} established a general LDP for conditional 
guesswork, but 
did not explicitly derive the exponent in the i.i.d.~case as we do 
below. 

Our type-counting proof strategy is to explicitly evaluate the probability 
that $\Xvec \sim P$ i.i.d.\ occurs in the $2^{nh}$ sequences with the 
highest probability.  Since the probability of such an $\Xvec$ only depends 
on its type $Q$, this is the same as checking that $Q$ is such that the 
total number of sequences among all types $Q' \better Q$ is smaller or 
greater than $2^{nh}$.  Finally, since the number of types is polynomial in 
$n$, the exponent of the probability is determined by the type $Q$ 
satisfying the above constraint that has the smallest $D(Q \| P)$ (recall 
that $\Pr\{\Xvec \in \typical{n}{Q}\} \doteq \exp(-n\, D(Q\| P))$).  To 
solve the resulting optimization problem, we show that the solution is 
unchanged if we instead consider a relaxed version of the same problem, 
whose solution is easily computed to be the tilted distribution 
$\tiltP{P}{\lambda_h}$ with tilting parameter $\lambda_h$ such that 
$H(\tiltP{P}{\lambda_h}) = h$.  Thus, the tilted distribution is exactly 
the type whose corresponding type class has the largest probability among 
all types that appear in the first $2^{nh}$ guesses.

\vspace*{5pt}
\begin{proposition}[Unconditional guesswork exponent] \label{prop: exponent}
	Let $h \in (0,\log |\calX|)$.     Then, the exponents 
	$e_X^{\gtrless}(h) \coloneqq \lim_{n\to\infty} 
	-\tfrac{1}{n}\log\Pr\{G(\Xvec) \gtrless 2^{nh}\}$ are respectively 
	given by
	\begin{equation*} 
		e_X^{\gtrless}(h) =  D(\tiltP{P}{\lambda_{h}} \,\|\, P) \cdot 
		\indic\{h \gtrless H(P)\},
	\end{equation*}
	where $\gtrless$ is either $>$ or $<$ and $\lambda_h$ is such that 
	$H(\tiltP{P}{\lambda_{h}}) = h$ in both cases.
\end{proposition}
\begin{proof}
	Note that the number of guesses $G(\xvec)$ can be upper and lower 
	bounded as
	\begin{equation*}
		\ell(Q_{\xvec}) \coloneqq \sum_{Q' \better Q_{\xvec}} 
		|\typical{n}{Q'}| \leq G(\xvec) \leq \sum_{Q' \bettereq Q_{\xvec}} 
		|\typical{n}{Q'}| \eqqcolon u(Q_{\xvec}).
	\end{equation*}
	Then, up to polynomial factors in $n$, we have, for $Q \in 
	\calP_n(\calX)$,  $\ell(Q) \doteq \exp[n\, \max_{Q' \better Q} H(Q')]$ 
	and $u(Q) \doteq \exp[n\, \max_{Q' \bettereq Q} H(Q')]$.  Further, the 
	probability that $\Xvec$ has type $Q$ is $\Pr\{\Xvec \in 
	\typical{n}{Q}\} \doteq \exp[-n\, D(Q \| P)]$.  Since we are interested 
	in the regime when $n \to \infty$, note that for $\ell(Q)$, the maximum 
	over the $Q'$'s satisfying the strict inequality becomes a supremum in 
	the limit, with the maximum attained in the closure, i.e., allowing $Q' 
	\bettereq Q$, and hence the exponents of both $\ell(Q)$ and $u(Q)$ are 
	$\max_{Q' \bettereq Q} H(Q')$.
	
	\sloppy
	Putting these together, we have that the exponent $e_X^{>}(h)$ is 
	determined by the type $Q$ that has the smallest $D(Q \| P)$ while 
	satisfying $\max_{Q' \bettereq Q} H(Q') \geq h$ (and similarly for 
	$e_X^{<}(h)$ and $\max_{Q' \bettereq Q} H(Q') \leq h$), i.e., 
	\begin{equation*}
		e_X^{\gtrless}(h) = \min_{Q : \max_{Q' \bettereq Q} H(Q') 
			\gtreqless h} D(Q \,\|\, P),
	\end{equation*}
	where $Q, Q' \in \calP(\calX)$ and $\gtreqless$ is $\geq$ or $\leq$ 
	accordingly when $\gtrless$ is $>$ or $<$.  In both cases, we show that 
	the above optimization problem is respectively equivalent to the 
	simpler problem
	\begin{equation*}
		E^{\gtrless}(h) \coloneqq \min_{Q : H(Q) \gtreqless h} D(Q \,\|\, 
		P).
	\end{equation*}
	
	We will show that $e_X^{\gtrless}(h) = E^{\gtrless}(h)$, but we
	first show that $E^{\gtrless}(h)$ is equal to $D(\tiltP{P}{\lambda_{h}} 
	\| P) \cdot \indic\{h \gtrless H(P)\}$, as desired.
	Clearly, if $H(P) \leq h$, then $E^{<}(h) = 0$, and if $H(P) \geq h$, 
	then $E^{>}(h) = 0$, as $Q = P$ is feasible.  
	
	We first consider $E^{>}(h) = \min_{Q : H(Q) \geq h} D(Q \| P)$ and 
	suppose 
	$h > H(P)$.  Then, we claim that the solution lies on the boundary $\{Q 
	: 
	H(Q) = h\}$.  To see why this is true, suppose the minimizer $Q^*$ has 
	$H(Q^*) > h$, and consider $Q_t = (1-t)Q^* + t P$ for $t \in (0,1)$.  
	Then, 
	by the convexity of divergence, $D(Q_t \| P) \leq (1-t) D(Q^* \| P) < 
	D(Q^* 
	\| P)$. Meanwhile, for sufficiently small $t > 0$, we still have 
	$H(Q_t) > 
	h$ as the entropy is continuous on finite sets, which contradicts the 
	optimality of $Q^*$.  Hence, $Q^*$ must satisfy $H(Q^*) = h$.  
	A completely identical argument also shows that the minimizer in 
	$E^{<}(h)$  for $h < H(P)$ also lies on the boundary  $\{Q : H(Q) = 
	h\}$. 
	
	Hence, it is enough to solve $\min_{Q : H(Q) = h} D(Q \| P)$.  We claim 
	that the minimizer is $Q = \tiltP{P}{\lambda_h}$, where $\lambda_h$ is 
	such that $H(\tiltP{P}{\lambda_h}) = h$.  Clearly, 
	$\tiltP{P}{\lambda_h}$ satisfies the constraint.  It is enough to show 
	that $D(Q \| P) \geq D(\tiltP{P}{\lambda_h} \| P)$ for all $Q$ with 
	$H(Q) = h$.  Consider the difference $D(Q \| P) - 
	D(\tiltP{P}{\lambda_h} \| P) = H(Q \| P) - H(\tiltP{P}{\lambda_h} \| P)$
	\begin{equation*}
		= \sum_x  \left(Q(x) - \tiltP{P}{\lambda_h}(x)\right) \log 
		\frac{1}{P(x)}.
	\end{equation*}
	Also consider the divergence $D(Q \| \tiltP{P}{\lambda_h}) = H(Q \| 
	\tiltP{P}{\lambda_h}) - H(Q)$, and since $H(\tiltP{P}{\lambda_h}) = 
	H(Q)$, we have that $D(Q \| \tiltP{P}{\lambda_h})$
	\begin{align*}
		&= H(Q \| \tiltP{P}{\lambda_h}) - H(\tiltP{P}{\lambda_h})\\
		&= \sum_x  \left(Q(x) - \tiltP{P}{\lambda_h}(x)\right) \log 
		\frac{1}{\tiltP{P}{\lambda_h}(x)} \\
		&= \lambda_h [D(Q\| P) - D(\tiltP{P}{\lambda_h} \| P)],
	\end{align*}
	from the above.  Hence, $D(Q\| P) - D(\tiltP{P}{\lambda_h} \| P) \geq 
	0$, and we have that $\tiltP{P}{\lambda_h}$ is the desired minimizer. 
	(We could have also used Lagrange multipliers to arrive at the same 
	minimizer.)  To complete the proof, it remains to show that 
	$e_X^{\gtrless}(h) = E^{\gtrless}(h)$.
	
	First consider $e_X^{>}(h) = \min_{Q : \max_{Q' \bettereq Q} H(Q') \geq 
		h} D(Q \| P)$.   Suppose, for a fixed $Q$ with $\max_{Q' \bettereq Q} 
	H(Q') \geq h$, that $\tilde Q$ (which depends on $Q$) is a maximizer of 
	$\max_{Q' \bettereq Q} H(Q')$.  Then we have $H(\tilde Q) \geq H(Q)$ 
	and $H(\tilde Q) \geq h$, as $\tilde Q$ is the maximizer, and $H(\tilde 
	Q) + D(\tilde Q \| P) \leq H(Q) + D(Q \| P)$, as $\tilde Q \bettereq 
	Q$.  Together, this implies that $D(\tilde Q \| P) \leq D(Q \| P)$.  
	
	Hence, for every $Q$ that is feasible in the original problem 
	(satisfying $\max_{Q' \bettereq Q} H(Q') \geq h$), we have a $\tilde Q$ 
	that satisfies $D(\tilde Q \| P) \leq D(Q \| P)$ and $H(\tilde Q) \geq 
	h$, i.e., 
	\begin{equation*}
		e_X^{>}(h) \geq \min_{Q : H(Q) \geq h} D(Q \| P) = E^{>}(h).
	\end{equation*}
	On the other hand, note that $Q$ with $H(Q) \geq h$ also necessarily 
	satisfies $\max_{Q' \bettereq Q} H(Q') \geq h$, and hence, the reverse 
	inequality, and consequently equality, holds.

	Now consider $e_X^{<}(h) = \min_{Q : \max_{Q' \bettereq Q} H(Q') \leq 
		h} D(Q \| P)$.  
	Note that we have $E^{<}(h) \leq e_X^{<}(h)$ in general, as $Q$ with 
	$\max_{Q' \bettereq Q} H(Q') \leq h$ also satisfies $H(Q) \leq h$.   
	Hence, if we can show that the minimizer $Q^*$ of the relaxed problem 
	$\min_{Q : H(Q) \leq h} D(Q \| P)$ satisfies the original constraint 
	$\max_{Q' \bettereq Q^*} H(Q') \leq h$, then we are done. 
	
	If $h \geq H(P)$, then $Q^* = P$, and for any $Q' \bettereq Q^*$,
	\begin{equation*}
		H(Q') \leq H(Q') + D(Q' \| P) \leq H(Q^*) + D(Q^* \| P),
	\end{equation*}
	which is equal to $H(P) \leq h$, and hence, $\max_{Q' \bettereq Q^*} 
	H(Q') \leq h$.
	On the other hand, for $h < H(P)$, then $Q^* = \tiltP{P}{\lambda_h}$ 
	with $H(Q^*)=h$.
	Note that $\max_{Q' \bettereq \tiltP{P}{\lambda}} H(Q')$ for $\lambda > 
	0$ 
	involves maximizing a concave function subject to a linear constraint. 
	Solving this via Lagrange multipliers yields the maximizer to be 
	exactly 
	$\tiltP{P}{\lambda}$, and hence, $\max_{Q' \bettereq Q^*} H(Q') = 
	H(\tiltP{P}{\lambda_h}) = h$, and we are done.
\end{proof}

\begin{remark}
	One way to get an easy upper bound to $\Pr\{G(\Xvec) > 2^{nh}\}$ is to 
	use Arikan's bound~\cite{arikan} on the moments of $G(\Xvec)$, i.e., 
	$\bbE[G(\Xvec)^{\rho}] \leq \left(\sum_x 
	P(x)^{\frac{1}{1+\rho}}\right)^{n(1+\rho)}$ for $\rho > 0$, together 
	with Markov's inequality
	\begin{align*}
		\Pr\{G(\Xvec) > 2^{nh}\} \leq 2^{-nh\rho} \bbE[G(\Xvec)^{\rho}].
	\end{align*}
	Optimizing over the choice of $\rho$ gives the best bound at $\rho = 
	1/\lambda_h - 1$, with $\lambda_h \in (0,1)$ as above (which is 
	possible if $h > H(P)$), and the exponent is exactly 
	$D(\tiltP{P}{\lambda_h} \| P)$.
\end{remark}

One might naively expect a conditional version of the above statement to 
follow immediately by simply replacing $P$ by $P_{X|Y}$, but it is easy to 
see that this is not the case, as the $X_i$ are not identically distributed 
given $\Yvec = \yvec$.  Indeed, $X_i$ has distribution $P_{X|Y}(\cdot | 
y_i)$ given $\Yvec = \yvec$.  Nonetheless, by repeating the type-counting 
argument above with correct order on \emph{conditional} types, we can 
compute the tail exponents of conditional guesswork.  We compute two such 
conditional exponents, namely that of $\Pr\{G(\Xvec|\Yvec)\gtrless 
2^{nh}\mid \Yvec=\yvec\}$ for $\yvec$'s whose type converges to a limit in 
Theorem~\ref{thm: single_cond_exponent}, and $\Pr\{G(\Xvec|\Yvec)\gtrless 
2^{nh}\}$ averaged over $\Yvec$ in Theorem~\ref{thm: cond_exponent}. 

\sloppy
\vspace*{3pt}
\begin{theorem}[Conditional guesswork exponent for given $\yvec$-type] 
	\label{thm: single_cond_exponent}
	Let $h \in \left(0, \sum_y Q(y)\log |\supp(P_{X|Y}(\cdot|y))|\right)$ 
	and $Q \in \calP(\calY)$. Then, for $\yvec^{(n)}$ such that the type 
	$Q_{\yvec}^{(n)} \to Q$, the exponents  $e_{X|Y}^{\gtrless}(h,Q) 
	\coloneqq \lim_{n\to\infty} -\tfrac1n \log \Pr\{G(\Xvec|\Yvec) \gtrless 
	2^{nh} \mid \Yvec = \yvec^{(n)}\}$ are respectively given by
	\begin{equation*}
		e_{X|Y}^{\gtrless}(h,Q)  =  D(\tiltP{P_{X|Y}}{\lambda_{h,Q}} \| 
		P_{X|Y} | Q) \cdot \indic\{h \gtrless H_{QP_{X|Y}}(X | Y)\},
	\end{equation*}
	where $\gtrless$ is either $>$ or $<$ and $\lambda_{h,Q}$ is such that 
	$H_{Q\tiltP{P_{X|Y}}{\lambda_{h,Q}}}(X | Y) = h$ in both cases. 
\end{theorem}
\begin{proof}[Proof sketch]
	The proof follows identically to that of Proposition~\ref{prop: 
		exponent}, except that we replace the initial bounds on 
	$G(\xvec|\yvec)$ by terms involving conditional distributions and 
	conditional type classes, i.e., for $\xvec \in 
	\typical{n}{Q_{X|Y}}(\yvec)$, we have
	\begin{equation*}
		\ell(Q_{X|Y},\yvec)  \leq G(\xvec|\yvec) \leq u(Q_{X|Y}, \yvec),
	\end{equation*}
	where $\ell(V, \yvec) \coloneqq \sum_{V' \better V} 
	|\typical{n}{V'}(\yvec)|$ and $u(V, \yvec)$ is identical with $\better$ 
	replaced by $\bettereq$, with both $V,V'$ being conditional types.   As 
	before, we have $\ell(V,\yvec) \doteq \exp[n\, \max_{V' \better V} 
	H_{Q_{\yvec}V'}(X | Y)]$, and $u(V,\yvec)$ is identical with 
	$\bettereq$ instead of $\better$. 
	
	Together with the probability of $\Xvec$ having conditional type $V$ 
	given $\Yvec = \yvec$ being $\Pr\{\Xvec \in \typical{n}{V}(\yvec) \mid 
	\Yvec = \yvec\} \doteq \exp[-n\,D(V \| P_{X|Y} | Q_{\yvec})]$, we have 
	that the exponent is
	\begin{equation*}
		e_{X|Y}^{\gtrless}(h,Q) = \min_{V : \max_{V' \bettereq V} H_{QV'}(X 
			| Y) \gtreqless h} D(V \,\|\, P_{X|Y} \mid Q)
	\end{equation*}
	with $V,V'$ going over all conditional types on $\calX$ from $\calY$.   
	By a completely analogous reasoning to the unconditional case, we claim 
	that the minimization problem is unchanged if we relax the constraint 
	on the conditional type $V$ to simply $H_{Q V}(X | Y) \gtreqless h$. 
	The solution to this simpler problem occurs at $V = P_{X|Y}$ for 
	$e^{>}_{X|Y}(h,Q)$ if $ H_{QP_{X|Y}}(X|Y) \geq h$ and for 
	$e^{<}_{X|Y}(h,Q)$ if $ H_{QP_{X|Y}}(X|Y) \leq h$.  On the other hand, 
	the solution is  at $V = \tiltP{P_{X|Y}}{\lambda_{h,Q}}$ for 
	$e^{>}_{X|Y}(h,Q)$ if $ H_{QP_{X|Y}}(X|Y) < h$ and for 
	$e^{<}_{X|Y}(h,Q)$ if $ H_{QP_{X|Y}}(X|Y) > h$.  Note that the tilted 
	distribution here is the conditional tilted distribution, which is 
	obtained by tilting $P_{X|Y}(\cdot|y)$ with parameter $\lambda_{h,Q}$ 
	for each $y$.
\end{proof}

\begin{theorem}[Conditional guesswork exponent] \label{thm: cond_exponent}
	Let $h \in \left(0, \max_y \log |\supp(P_{X|Y}(\cdot|y))|\right)$.
	For $\lambda > 0$, define
	\begin{equation}
		Q_{Y,\lambda}(y) \coloneqq \frac{P_Y(y)\, \left(\sum_x 
			P_{X|Y}(x|y)^{\lambda}\right)^{1/\lambda}}{\sum_{y'} P_Y(y')\, 
			\left(\sum_{x'} P_{X|Y}(x'|y')^{\lambda}\right)^{1/\lambda}}. 
		\label{eqn: q_cond_exp}
	\end{equation}
	Then, the exponents of $\Pr\{G(\Xvec|\Yvec) \gtrless 2^{nh}\}$ as 
	defined in \eqref{eqn: cond_exponents} are respectively given by
	\begin{equation*}
		e_{X|Y}^{\gtrless}(h)
		=
		D(Q_{Y,\lambda_h}\tiltP{P_{X|Y}}{\lambda_h} \,\|\, P_{XY})
		\cdot \indic\{h \gtrless H(X \mid Y)\},
	\end{equation*}
	where $\gtrless$ is either $>$ or $<$ and $\lambda_h$ is such that 
	$H_{Q_{Y,\lambda_h}\tiltP{P_{X|Y}}{\lambda_h}}(X \mid Y) = h$ in both 
	cases.
\end{theorem}
\begin{remark}
	Before proceeding to the proof, it is worth noting that if $P_{X|Y}$ is 
	symmetric in the sense that $\sum_x P_{X|Y}(x|y)^{\lambda_h}$ is a 
	constant for all $y$, then $Q_{Y,\lambda_h} = P_Y$, and the exponents 
	are $D(P_Y\tiltP{P_{X|Y}}{\lambda_h} \| P_{XY}) = 
	D(\tiltP{P_{X|Y}}{\lambda_h} \| P_{X|Y} | P_Y )$ which is equal to the 
	conditional exponent $e^{\gtrless}_{X|Y}(h, P_Y)$.   
\end{remark}

\begin{proof}[Proof sketch]
	We follow the same procedure as in the proof of Proposition~\ref{prop: 
		exponent}, except that we go over joint types $Q_{\xvec\yvec}$ to 
	compute the probability instead of types on $\calX$ only.  Recall that 
	the (partial) order on joint types is defined so that $Q_{XY} \bettereq 
	Q'_{XY}$ means that $Q_Y = Q'_Y$ and $Q_{X|Y} \bettereq Q'_{X|Y}$.  By 
	following the same steps as earlier, we have that the exponent is
	\begin{equation*}
		e_{X|Y}^{\gtrless}(h)
		=
		\min_{Q : \max_{Q' \bettereq Q} H_{Q'}(X|Y) \gtreqless h} D(Q 
		\,\|\, P_{XY}),
	\end{equation*}
	where $Q, Q' \in \calP(\calX \times \calY)$ are joint distributions, 
	and $\gtreqless$ is $\geq$ when $\gtrless$ is $>$ and $\leq$ when 
	$\gtrless$ is $<$.  By the same reasoning as in the unconditional case, 
	this minimization problem is unchanged if we simplify the constraint to 
	$H_Q(X|Y) \gtreqless h$, i.e.,
	\begin{equation*}
		e_{X|Y}^{\gtrless}(h)
		=
		\min_{Q : H_Q(X|Y) \gtreqless h} D(Q \,\|\, P_{XY}).
	\end{equation*}
	
	Clearly, if $H(X|Y) \gtreqless h$, then $Q = P_{XY}$ is feasible and 
	the exponent is $0$.  Hence, it remains to solve
	\begin{equation*}
		\min_{Q : H_Q(X|Y) = h} D(Q \,\|\, P_{XY})
	\end{equation*}
	in the non-trivial case $h \gtrless H(X|Y)$.  We can solve this using 
	Lagrange multipliers to obtain that the optimizer is of the form 
	$Q(x,y) = Q(y)V(x|y)$ where $V$ is the conditional tilted distribution 
	$\tiltP{P_{X|Y}}{\lambda_h}$ and $Q(y) \propto P_Y(y) \left(\sum_x 
	P_{X|Y}(x|y)^{\lambda_h}\right)^{1/\lambda_h}$, where $\lambda_h$ is 
	such that the conditional entropy $H(X|Y)$ with this $Q$ is $h$.
\end{proof}

\begin{remark}
	A more illustrative proof of Theorem~\ref{thm: cond_exponent} can 
	perhaps be obtained using Theorem~\ref{thm: single_cond_exponent} by 
	simply averaging over $\Yvec = \yvec$, as $\Pr\{G(\Xvec|\Yvec) \gtrless 
	2^{nh}\}$
	\begin{equation*}
		= \sum_{\yvec} P_{\Yvec}(\yvec)\Pr\{G(\Xvec|\Yvec) \gtrless 2^{nh} 
		| \Yvec =\yvec\}.\vspace*{-7pt}
	\end{equation*}
	Since the terms in the sum only depend on the type of $\yvec$, we can 
	sum over all possible types $Q \in \calP_n(\calY)$ instead, with each 
	type having exponent $D(Q\|P_Y) + e^{\gtrless}_{X|Y}(h,Q) = $
	\begin{equation*}
		D(Q\|P_Y) + D(\tiltP{P_{X|Y}}{\lambda_{h,Q}} \| P_{X|Y} \mid Q) = 
		D(Q\tiltP{P_{X|Y}}{\lambda_{h,Q}} \| P_{XY}),
	\end{equation*}
	where $\lambda_{h,Q}$ satisfies 
	$H_{Q\tiltP{P_{X|Y}}{\lambda_{h,Q}}}(X|Y) = h$ as before. 
	Hence, the averaged conditional exponent is
	\begin{equation*}
		e^{\gtrless}_{X|Y}(h) = \min_{Q \in \calP(\calY)}\left[ D(Q\|P_Y) + 
		e^{\gtrless}_{X|Y}(h,Q) \right],
	\end{equation*}
	However, there does not seem to be an easier way to see that this 
	exponent is minimized at $Q = Q_{Y,\lambda_h}$ as in \eqref{eqn: 
		q_cond_exp}, other than resorting to Lagrange multipliers. 
	
	Further, note that for the averaged conditional exponent in 
	Theorem~\ref{thm: cond_exponent}, it is not a tilted joint distribution 
	$\tiltP{P_{XY}}{\lambda}$ for some appropriate $\lambda$ that 
	determines 
	the exponent, but rather, the joint distribution induced by the 
	conditional 
	tilted distribution $\tiltP{P_{X|Y}}{\lambda_h}$ together with some 
	marginal distribution $Q_{Y,\lambda_h}$. Thus, the important geodesic 
	here 
	is the family of conditional tilted distributions, similar to the 
	tilted 
	distributions in the unconditional case (Proposition~\ref{prop: 
		exponent}).
\end{remark}
\begin{remark}
	Note that the moment characterization tells us that 
	$\bbE[G(\Xvec|\Yvec)]$ grows exponentially in $n$ with the exponent 
	related to the R\'enyi entropy of order $1/2$. Nonetheless, in the 
	results above, we see that it is the Shannon entropy (which is the 
	R\'enyi entropy of order $1$) that marks the transition point: 
	$\Pr\{G(\Xvec|\Yvec) < 2^{nh}\}$ goes to 0 if $h < H(X|Y)$ and goes to 
	1 otherwise.  This is not surprising, as it can be shown that 
	$\lim_{n\to\infty}\frac1n\bbE[\log G(\Xvec|\Yvec)] = H(X|Y)$ (for 
	example, by taking $\rho \to 0$ in the moment characterization). 
\end{remark}

\section{Application to password guessing}
We now use Theorem~\ref{thm: cond_exponent} to obtain a quantitative 
estimate for a natural problem in security.
Suppose we have a password $\Xvec$ that is an $n$-length string generated 
i.i.d.~from distribution $P_X$, and suppose an attacker observes 
side-information $\Yvec$ through a memoryless channel $P_{Y|X}$. For 
example, $P_X$ is a Zipf law~\cite{zipf} on an alphabet $\{1,\dots,m\}$, 
i.e., $P_X(x) = x^{-s}/\sum_{x'} x'^{-s}$ with $s > 0$, and $P_{Y|X}$ is an 
$m$-ary erasure channel, i.e, $P_{Y|X}(y|x) = 1-\epsilon$ for $y = x$ and 
$\epsilon$ for $y$ equal to a special ``erasure'' symbol not in $\calX$.  
If our system is such that it allows a small, finite number of attempts at 
entering the password before locking out the user, then we wish to ensure 
that the attacker is, with high probability, unable to guess the password 
in a constant number of attempts.  For $\delta \to 0$, if we choose $n 
\gtrsim \log\frac{1}{\delta} / e^{<}_{X|Y}(\eta)$, then, we have that 
$\Pr\{G(\Xvec | \Yvec) < t_n \} < \delta$ for $t_n$ that grows slower than 
$2^{n\eta}$ for $\eta > 0$. 
Note that $e^{<}_{X|Y}(\eta) = 
D(Q_{Y,\lambda_{\eta}}\tiltP{P_{X|Y}}{\lambda_{\eta}} \| P_{XY})$, where 
$\lambda_{\eta} \to \infty$ as $\eta \to 0$.  Assuming that $\max_x 
P_{X|Y}(x|y)$ is unique, we can check that 
\begin{equation*}
	e^{<}_{X|Y}(\eta) \stackrel{\eta\to0}{\longrightarrow} -\log  
	\left(\sum_y P_Y(y) \max_x P_{X|Y}(x|y) \right),
\end{equation*}
which is exactly the conditional min-entropy~\cite{minentropy}.
For the Zipf source and erasure channel setting considered above, this 
expression evaluates to $-\log\left((1-\epsilon) + \epsilon/\sum_{i=1}^m 
i^{-s}\right)$.  To give a concrete example, let $\delta = 10^{-6}$, $s = 
0.8$, $m = 100$ and $\epsilon = 0.5$, then this exponent is approximately 
$0.83$.  We need $n \gtrsim \frac{\log(1/\delta)}{0.83} \approx 24$ to 
ensure that even if around half of the (i.i.d.) characters in the password 
are leaked, the attacker can succeed in guessing the right password in a 
small number of attempts only with probability $10^{-6}$. 

\section{Conclusion}
We give a simple method-of-types derivation of the tail exponents of 
conditional guesswork for i.i.d.~sources.  The exponent is characterized as 
the 
solution to an entropy-constrained divergence minimization problem, which 
yields a divergence involving conditional tilted distributions.  This lets 
us 
easily compute the exponent for use in applications, as we illustrate in a 
password guessing example.  Though our results only show asymptotic 
exponents, 
the method-of-types derivation can also give upper and lower bounds to the 
probability that match up to polynomial factors in $n$.  Another possible 
application is the analysis of guessing-based 
decoding~\cite{duffy_grand,duffy_srgrand,duffy_orbgrand,miyamoto_yang_universal,li_zhang_orbgrand},
where we are interested in quantifying the probability that the true noise 
sequence is not reached within a prescribed number of guesses.


\printbibliography

\end{document}